\documentclass{llncs}

\usepackage{url}
\usepackage[pdftex]{graphicx}
\usepackage{amsmath,amssymb,xcolor,wrapfig,paralist}
\usepackage{algorithmicx,algorithm}
\usepackage[noend]{algpseudocode}
\usepackage{multirow}

\usepackage{booktabs}
\usepackage{proof}
\usepackage{tikz}
\usepackage{colortbl}
\usepackage{pdflscape}
\usetikzlibrary{shapes,positioning,arrows,calc,automata,matrix}
\tikzstyle{state}+=[minimum size = 8mm, inner sep=0,outer sep=1]
\tikzset{->,>=stealth'}

\definecolor{wwhite}{gray}{1}
\usepackage[T1]{fontenc}

\usepackage{tabularx}
\newcolumntype{L}{>{\raggedright\arraybackslash}p{1.6cm}}
\newcolumntype{C}{>{\centering\arraybackslash}p{1.6cm}}
\newcolumntype{R}[1]{>{\raggedleft\arraybackslash}p{#1}}
\allowdisplaybreaks

\newcommand{\thmhelperpre}[2]{\newcommand{\theoremlike}[1]{\par\medskip\penalty-250\refstepcounter{theorem}{\bfseries\noindent##1 \ref{#1}.}\itshape}\theoremlike{#2}}
\newcommand{\thmhelperpost}{\par\medskip%
 \renewcommand{\theoremlike}[1]{\par\medskip\penalty-250\refstepcounter{theorem}{\bfseries\noindent##1 \thesection .\thetheorem.}\itshape}%
}

\newcommand{\rc}[1]{#1}%{\textcolor{red}{#1}}

\newcommand{\pr}{\mathbb P}
\newcommand{\expected}{\mathbb E}

\newcommand{\pmin}{\mathsf{p}_{\mathsf{min}}}
\newcommand{\pvarphi}{\mathsf{p}_{\varphi}}

\newcommand{\Nset}{\mathbb N}
\newcommand{\trans}[1]{\xrightarrow{#1}}       %% Transition relation

\newcommand{\Mc}{\mathcal{M}}
\newcommand{\Pm}{\mathbf{P}}
\newcommand{\St}{S}
\newcommand{\init}{\mu}
\newcommand{\Lab}{L}

\newcommand{\dra}{\mathcal{A}}
\newcommand{\Mcdra}{\Mc \otimes \dra}
\newcommand{\draS}{Q}
\newcommand{\draAl}{{2^{Ap}}}
\newcommand{\draTr}{\gamma}
\newcommand{\draInit}{q_o}
\newcommand{\draAcc}{Acc}

\newcommand{\scc}{\mathsf{SCC}}
\newcommand{\bscc}{\mathsf{BSCC}}

\newcommand{\scs}{\mathsf{SC}}
\newcommand{\run}{\rho}
\renewcommand{\path}{\pi}
\newcommand{\emptypath}{\lambda}
\newcommand{\concat}{\,.\,}
\newcommand{\runs}{\mathsf{Runs}}
\newcommand{\candidatepath}{\kappa}
\newcommand{\candidate}{K}
\newcommand{\strength}{\textsc{Str}}
\newcommand{\kcand}[1]{\mathit{WCand}_{#1}}
\newcommand{\skcand}[1]{\mathit{Cand}_{#1}}
\newcommand{\support}[1]{\overline{#1}}

\newcommand{\bad}{\textsc{Bad}}
\newcommand{\tup}[1]{\langle #1 \rangle}

\newcommand{\F}{{\ensuremath{\mathbf{F}}}}
\newcommand{\G}{{\ensuremath{\mathbf{G}}}}

\newcommand{\cone}{\mathsf{Cone}}
\newcommand{\nextState}{\mathsf{NextState}}

\newcommand{\trerr}[1]{\xi}
\newcommand{\mperr}[1]{\zeta}

\newcommand{\maxsizescc}{\mathsf{{\scriptsize mxsc}}}

\newcommand{\rem}[1]{}

\newcommand{\mybirthday}{\mathit{Birthday}}
\newcommand{\myvisits}{\mathit{Visits}}

\newcommand{\etal}{\textit{et al.\ }}

\newtheorem{fact}{Fact}

\usepackage[T1]{fontenc}
\usepackage{lmodern}
\usepackage{microtype}
\newcommand{\myspace}{\vspace*{-0.5em}}

 \advance\textwidth2mm % <= 5mm
 \advance\hoffset-1mm
\title{Online Monitoring $\omega$-Regular Properties in Unknown Markov Chains}
\author{Javier Esparza\inst{1} \and
Stefan Kiefer\inst{2} \and
Jan K\v{r}et{\'i}nsk{\'y} \inst{1}\and
Maximilian Weininger\inst{1}}
\institute{
Technische Universit\"at M\"unchen \and University of Oxford
}

\begin{document}

\pagestyle{plain}
\maketitle

\begin{abstract}
We study runtime monitoring of $\omega$-regular properties. We consider a simple setting in which  a run of an unknown finite-state Markov chain $\Mc$ is monitored against a fixed but arbitrary $\omega$-regular specification $\varphi$. The purpose of monitoring is to keep aborting runs that are ``unlikely'' to satisfy the specification until $\Mc$ executes  a correct run. We design controllers for the reset action that (assuming that $\varphi$ has positive probability) satisfy the following property w.p.1: the number of resets is finite, and the run executed by $\Mc$ after the last reset satisfies $\varphi$.
\end{abstract}

\section{Introduction}
Runtime verification, also called runtime monitoring, is the problem of checking at runtime whether an execution of a system satisfies a given correctness property (see e.g. \cite{HavelundR04a,LeuckerS09,FalconeHR13,BartocciFFR18}). It can be used to automatically evaluate test runs, or to steer the application back to some safety region if a property is violated. Runtime verification of LTL or $\omega$-regular properties has been thoroughly studied \cite{BauerLS06,LeuckerS09,BauerLS11,BartocciBNR18}. It is conducted by automatically translating the property into a monitor that inspects the execution online in an incremental way, and (in the most basic setting) outputs ``yes'', ``no'', or ``unknown'' after each step. 
%An LTL property is \emph{monitorable} if there exists a monitor that eventually outputs ``yes'' or ``no'' \cite{}. 
A fundamental limitation of runtime verification is that, if the system is not known \textit{a priori}, then many properties, like for example $\G\F p$ or $\F\G p$, are not monitorable. Loosely speaking, since every finite execution can be extended to a run satisfying $\G\F p$ and to  another run satisfying its negation, monitors can only continuously answer ``unknown'' (see \cite{BartocciBNR18} for a more detailed discussion). Several approaches to this problem have been presented, which modify the semantics of LTL in different ways to refine the prediction and palliate the problem \cite{PnueliZ06,BauerLS07,BauerLS10,MorgensternGS12,ZhangLD12,BartocciBNR18}, but the problem is of fundamental nature.

Runtime monitoring of stochastic systems modeled as Hidden Markov Chains (HMM) has been studied by Sistla \etal in a number of papers \cite{SistlaS08,GondiPS09,SistlaZF11}. Given a HMM $H$ and an $\omega$-regular language $L$, % and an arbitrary $0 \leq  p < 1$, 
these works construct a monitor that (a) rejects executions of $H$ not in $L$ w.p.1, and (b) accepts executions of $H$ in $L$ with positive probability (this is called a \emph{strong monitor} in \cite{SistlaS08}). Observe, however, that the monitor knows $H$ in advance. The case where $H$ is not known in advance is also considered in \cite{SistlaS08}, but in this case strong monitors only exist for languages recognizable by deterministic, possibly infinite state, B\"{u}chi automata. Indeed, it is easy to see that, for example, a monitor that has no information about the HMM cannot be strong for a property like $\G\F p$.

Summarizing, the work of Sistla \etal seems to indicate that one must either know the HMM in advance, or has to give up monitorability of liveness properties.  In this paper we leverage a technique introduced in \cite{DacaHKP17} to show that there is a third way: Assume that, instead of only being able to observe the output of a state, as in the case of HMM, we can observe the state itself. In particular, we can observe that the current state is the same we visited at some earlier point.  We show that this allows us to design monitors for all $\omega$-regular properties that work without any knowledge of the system in the following simple setting. We have a finite-state Markov chain, but we have no information on its size, probabilities, or structure; we can only execute it. We are also given an arbitrary $\omega$-regular property $\varphi$, and the purpose of monitoring is to abort runs of the system that are ``unlikely'' to satisfy the specification until the system executes  a correct run. Let us make this informal idea more precise. The semantics of the system is a finite-state Markov chain, but we have no information on its size, probabilities, or structure. We know that the runs satisfying $\varphi$ have nonzero probability, but the probability is also unknown.  We are allowed to monitor runs of the system and record the sequence of states it visits; further, we are allowed to abort the current run at any moment in time, and reset the system back to its initial state. The challenge is to design a controller for the reset action that satisfies the following property w.p.1: the number of resets is finite, and the run of the system after the last reset satisfies $\varphi$.  

Intuitively, the controller must abort the right number of executions: If it aborts too many, then it may reset infinitely often with positive probability; if it aborts too few, the run after the last reset might violate $\varphi$.
For a safety property like $\G p$ the controller can just abort whenever the current state does not satisfy $p$; indeed, since $\G p$ has positive probability by assumption, eventually the chain executes a run satisfying $\G p$ a.s., and this run is not aborted.  Similarly, for a co-safety property like $\F p$, the controller can abort the first execution after one step, the second after two, steps etc., until a state satisfying $p$ is reached. Since $\F p$ has positive probability by assumption, at least one reachable state satisfies $p$, and with this strategy the system will almost surely visit it.
But for $\G\F p$ the problem is already more challenging. Unlike the cases of $\G p$ and $\F p$, the controller can never be sure that every extension of the current execution will satisfy the property or will violate it. 

In our first result we show that, perhaps surprisingly, notions introduced in \cite{DacaHKP17} can be used to show that the problem has  very simple solution. Let $\mathcal{M}$ be the (unknown) Markov chain of the system, and let $\mathcal{A}$ be a deterministic Rabin automaton for $\varphi$.  Say that a run of the product chain $\mathcal{M} \times \mathcal{A}$ is \emph{good} if it satisfies $\varphi$, and \emph{bad} otherwise.
We define a set of \emph{suspect} finite executions satisfying two properties:
\begin{itemize}
\item[(a)] bad runs a.s. have a suspect prefix; and
\item[(b)] if the set of good runs has nonzero probability, then the set of runs without suspect prefixes also has nonzero probability.
\end{itemize}
The controller resets whenever the current execution is suspect. We call it the \emph{cautious controller}. By property (a) the cautious controller aborts bad runs w.p.1, and by property (b) w.p.1 the system eventually executes a run without suspect prefixes, which by (a) is necessarily good.

The performance of a controller is naturally measured in terms of two parameters: the expected number of resets $R$, and the expected number $S$ of steps to a reset (conditioned on the occurrence of the reset). While the cautious controller is very simple, it has poor performance: In the worst case, both parameters are exponential in the number of states of the chain.  A simple analysis shows that, without further information on the chain, the exponential dependence in $S$ is unavoidable. However, the exponential dependence on $R$ can be avoided: using a technique of \cite{DacaHKP17}, we define a  \emph{bold controller} for which the expected number of resets is almost optimal. 

\paragraph{Related work.} Sistla \etal have also studied the power of finite-state probabilistic monitors for the analysis of non-probabilistic systems, and characterized the monitorable properties \cite{ChadhaSV09}. This is also connected to work by Baier \etal \cite{BaierGB12}. There is also a lot of work on the design of monitors whose purpose is not to abort runs that violate a property, say $\varphi$, but gain information about the probability of the runs that satisfy $\varphi$. This is often called statistical model checking, and we refer the reader to  \cite{LegayDB10} for an overview.

\paragraph{Appendix.} Some proofs have been moved to an Appendix available  at \\
\url{https://www7.in.tum.de/~esparza/tacas2021-134.pdf} 

\section{Preliminaries}

\paragraph{Directed graphs.} A directed graph is a pair $G=(V, E)$, where $V$ is the set of vertices and $E \subseteq V\times V$ is the set of edges. A path (infinite path) of $G$ is a finite (infinite) sequence $\path = v_0, v_1 \ldots$ of vertices such that $(v_i, v_{i+1}) \in E$ for every $i=0,1 \ldots$.  We denote the empty path by $\emptypath$ and concatenation of paths $\path_1$ and $\path_2$ by $\path_1\concat\path_2$. A graph $G$ is strongly connected if for every two vertices $v, v'$ there is a path leading from $v$ to $v'$. A graph $G'=(V',E')$ is a subgraph of $G$, denoted $G' \preceq G$, if $V' \subseteq V$ and $E' \subseteq E \cap V' \times V'$; we write $G' \prec G$ if $G' \preceq G$ and $G'\neq G$. A graph $G' \preceq G$ is a strongly connected component (SCC) of $G$ if it is strongly connected and no graph $G''$ satisfying $G' \prec G'' \preceq G$ is strongly connected. An SCC $G'=(V',E')$ of $G$  is a bottom SCC (BSCC) if $v \in V'$ and $(v, v') \in E$ imply $v' \in V'$.

\paragraph{Markov chains.} A \emph{Markov chain (MC)} is a tuple $\Mc = (\St, \Pm, \init)$, where
\begin{itemize}
\item $\St$ is a finite set of states,
\item $\Pm \;:\; \St \times \St \to [0,1]$ is the transition probability matrix, such that for every $s\in \St$ it holds $\sum_{s'\in \St} \Pm(s,s') = 1$,
\item $\init$ is a probability distribution over $\St$. %\todotanja{or we enforce the initial state?}
\end{itemize}
The graph of $\Mc$ has $\St$ as vertices and $\{ (s, s') \mid \Pm(s,s') > 0\}$ as edges. Abusing language,
we also use $\Mc$ to denote the graph of $\Mc$.
We let $\pmin:=\min(\{\Pm(s,s') > 0\mid s,s'\in\St\})$ denote the smallest positive transition probability in $\Mc$.
% From now on, we will refer to $\Mc = (\St,\Pm,\init)$, unless it is specified otherwise.
A \emph{run} of $\Mc$ is an infinite path $\run = s_0 s_1 \cdots$ of $\Mc$; we let $\run[i]$ denote the state $s_i$. 
Each path $\path$ in $\Mc$ determines the set of runs $\cone(\path)$ consisting of all runs that start with $\path$. 
To $\Mc$ we assign the probability space $%\mathcal P_{\Mc}=
(\runs,\mathcal F,\pr)$, where $\runs$ is the set of all runs in $\Mc$, $\mathcal F$ is the $\sigma$-algebra generated by all $\mathsf{Cone}(\path)$,
and $\pr$ is the unique probability measure such that
$\pr[\mathsf{Cone}(s_0s_1\cdots s_k)] = 
\mu(s_0)\cdot\prod_{i=1}^{k} \Pm(s_{i-1},s_i)$, where the empty product equals $1$.
%Phrases ``almost surely'' or ``almost all runs'' refers to happening with probability 1 according to this measure.
The respective expected value of a random variable $f:\runs\to\mathbb R$ is $\expected[f]=\int_\runs f\ d\,\pr$.

Given a finite set $Ap$ of atomic propositions,
a \emph{labelled Markov chain} (LMC) is a tuple $\Mc = (\St, \Pm, \init, Ap, \Lab)$, where $(\St,\Pm, \init)$ is a MC and $\Lab : \St \to 2^{Ap}$ is a labelling function. 
Given a labelled Markov chain $\Mc$ and an LTL formula $\varphi$, we are interested in the measure $\pr[\Mc \models \varphi] := \pr[\{\run\in\runs \mid \Lab(\run) \models \varphi\}],$
%of paths $\run$ with the $\Lab$-projection satisfying $\varphi$. 
where $\Lab$ is naturally extended to runs by $\Lab(\run)[i]=\Lab(\run[i])$ for all $i$.  

\paragraph{Deterministic Rabin Automata.} For every $\omega$-regular property $\varphi$ there is a \emph{ deterministic Rabin automaton} (DRA) $\dra = (\draS, \draAl, \draTr, \draInit, \draAcc)$ that accepts all runs that satisfy $\varphi$~\cite{BK08}. 
Here
$\draS$ is a finite set of states, $\draTr : \draS \times \draAl \to \draS$ is the transition function, $\draInit \in \draS$ is the initial state, and $\draAcc \subseteq 2^\draS \times 2^\draS$ is the acceptance condition.

\paragraph{Product Markov Chain.} The product of %$\Mc = (\St, \Pm, \init, Ap, \Lab)$ and a DRA $\mathcal{A} = (\draS, \draAl, \draTr, \draInit, \draAcc)$, their product 
a MC $\Mc$ and DRA $\mathcal A$ is the Markov chain $\Mcdra =(\St \times \draS, \Pm', \init')$, where $\Pm'((s,q),(s',q')) = \Pm(s,s')$ if $q'=\draTr(q,\Lab(s'))$ and $\Pm'((s,q),(s',q'))=0$ otherwise,  and $\init'(s,q) = \init(s)$ if $\draTr(\draInit, \Lab(s))=q$ and $\init'(s,q)=0$ otherwise. 
%Note that $\Mcdra$ has the same smallest transition probability $\pmin$ as $\mathcal{M}$.

An SCC $B$ of $\Mcdra$ is \emph{good} if 
there exists a Rabin pair $(E,F) \in \draAcc$ such that $B \cap (S\times E) = \emptyset$ and $B \cap (S\times F) \neq \emptyset$. Otherwise, the SCC is \emph{bad}. Observe that the runs of $\Mcdra$ satisfying $\varphi$ almost surely reach a good BSCC, and the runs that do not satisfy $\varphi$ almost surely reach a bad BSCC.

\section{The cautious monitor} \label{sec:cautious}

All our monitors assume the existence of  a deterministic Rabin automaton $\dra = (\draS, \draAl, \draTr, \draInit, \draAcc)$ for $\varphi$.  They monitor the path $\path$ of the chain $\mathcal{M}\otimes \dra$ corresponding to the path of $\mathcal{M}$ executed so far. In order to present the cautious monitor we need some definitions and notations.

\paragraph{Candidate of a path.}  
Given a finite or infinite path $\run=s_0 s_1\cdots$ of $\mathcal{M}\otimes \dra$, the \emph{support} of $\run$ is the set $\support{\run}=\{s_0,s_1,\ldots\}$. The \emph{graph of $\run$} is $G_\run = (\support{\run}, E_\run)$, where $E_\run=\{(s_i,s_{i+1})\mid i=0,1,\ldots\}$. % for each $i$. 

Let $\path$ be a path of $\mathcal{M}\otimes \dra$. If 
$\path$ has a suffix $\candidatepath$ such that $G_\candidatepath$ is a BSCC of $G_\path$, we call $\support{\candidatepath}$ the \emph{candidate of $\path$}. Given a path $\path$, we define $\candidate(\path)$ as follows: If $\path$ has a candidate $\support{\candidatepath}$, then $\candidate(\path) := \support{\candidatepath}$; otherwise, $\candidate(\path) := \bot$, meaning that $\candidate(\path)$ is undefined. 

\begin{example}
Consider the family of Markov chains of Figure \ref{fig:largescc}.  We have e.g. $K(s_0) = K(s_0s_1) = K(s_0s_0s_1) = \bot$, $K(s_0 s_0) = \{s_0\}$, and $K(s_0s_1s_0s_1) = \{s_0,s_1\}$. In the family of Figure \ref{fig:smallscc} we have e.g. $K(s_0s_1s_1)=\{s_1\}$, $K(s_0s_1s_1s_2)= \bot $, and $K(s_0s_1s_1s_2s_2)= \{s_2\}$.
\end{example}

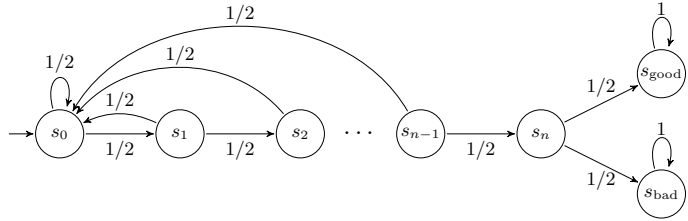
\begin{figure}
	\centering
	\scalebox{0.8}{
		\begin{tikzpicture}
		\node[state,initial,initial text=] (s0) at (0,1){$s_0$};
		\node[state] (s1) at (2,1){$s_1$};
		\node[state] (s2) at (4,1){$s_2$};
		\node (dots) at (5,1){\large $\cdots$};
		\node[state] (s3) at (6,1){$s_{n-1}$};
		\node[state] (s4) at (8,1){$s_n$};
		\node[state] (sg) at (10,2){$s_{\text{good}}$};
		\node[state] (sb) at (10,0){$s_{\text{bad}}$};
		\path[->] 
        (s0) edge[loop above] node[above=-1pt]{$1/2$} ()
		(s0) edge node[below]{$1/2$} (s1)
		(s1) edge node[below]{$1/2$} (s2)
		(s3) edge node[below]{$1/2$} (s4)
		(s1) edge[bend right=25] node[above=-2pt]{$1/2$} (s0)
		(s2) edge[bend right=50] node[above=-2pt]{$1/2$} (s0)
		(s3) edge[bend right=60] node[above=-2pt]{$1/2$} (s0)
		(sg) edge[loop above] node[above=-2pt]{$1$} ()
		(sb) edge[loop above] node[above=-2pt]{$1$} ()
		(s4) edge node[above]{$1/2$} (sg)
		(s4) edge node[below]{$1/2$} (sb);
		\end{tikzpicture}}
	\caption{A family of Markov chains}
	\label{fig:largescc}
\end{figure}

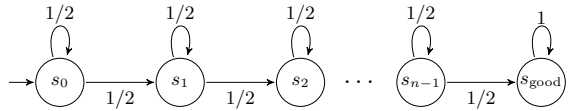
\begin{figure}[ht]
	\centering
	\scalebox{0.8}{
		\begin{tikzpicture}
		\node[state,initial,initial text=] (s0) at (0,1){$s_0$};
		\node[state] (s1) at (2,1){$s_1$};
		\node[state] (s2) at (4,1){$s_2$};
		\node (dots) at (5,1){\large $\cdots$};
		\node[state] (s3) at (6,1){$s_{n-1}$};
		\node[state] (sg) at (8,1){$s_{\text{good}}$};
		\path[->] 
		(s0) edge node[below]{$1/2$} (s1)
		(s1) edge node[below]{$1/2$} (s2)
		(s3) edge node[below]{$1/2$} (sg)
		(s0) edge[loop above] node[above=-2pt]{$1/2$} ()
		(s1) edge[loop above] node[above=-2pt]{$1/2$} ()
		(s2) edge[loop above] node[above=-2pt]{$1/2$} ()
		(s3) edge[loop above] node[above=-2pt]{$1/2$} ()
		(sg) edge[loop above] node[above=-2pt]{$1$} ();		
		\end{tikzpicture}}
	\caption{A family of Markov chains with small SCCs}
	\label{fig:smallscc}
\end{figure}

\paragraph{Good and bad candidates.} A candidate $\candidate$ is \emph{good} if there exists a Rabin pair $(E,F) \in \draAcc$ such that $\candidate \cap (S\times E) = \emptyset$ and $\candidate_i \cap (S\times F) \neq \emptyset$.   Otherwise, $\candidate$ is \emph{bad}.
A path $\path$ of $\mathcal{M}\otimes \dra$ is \emph{bad} if $\candidate(\path)\neq \bot$ and $\candidate(\path)$ is a bad candidate.
The function $\textsc{Bad}(\path)$ returns \textbf{true} if $\path$ is bad, and \textbf{false} otherwise. 

%\begin{definition}
%A path $\path$ of $\Mcdra$ is \emph{suspect} if it is bad and $\candidate(\path)$ has strength at least 2.
%\end{definition}

\begin{proposition}
\label{prop:suspect}
\begin{itemize}
\item[(a)] Bad runs of $\Mcdra$ almost surely have a bad finite prefix.
\item[(b)] If the good runs of $\Mcdra$ have nonzero probability, then the set of runs without bad prefixes also has nonzero probability.
\end{itemize}
\end{proposition}
\begin{proof}
(a) By standard properties of Markov chains, bad runs of $\Mcdra$ almost surely reach a BSCC of $\Mcdra$ and then traverse all edges of that BSCC infinitely often.
Therefore, a bad run~$\run$ almost surely has a finite prefix~$\path$ that has reached a bad BSCC, say $B$, of~$\Mcdra$ and has traversed all edges of~$B$ at least once.
Then $\candidate(\path) = B$, and so $\path$ is bad.
%Let $\run$ be a bad run of $\Mcdra$.
%Almost surely, $\run$ reaches a BSCC $S$ of the graph of $\Mcdra$
%W.p.1 the final candidate $\candidate$ of $\run$ is a BSCC of the graph of $\Mcdra$, and for every $k \geq 1$ 
%there is a prefix $\path_k$ of $\run$ such that $\candidate(\path_k) = \candidate$ and $\candidate(\path_k)$ has strength at least $k$. 
%Since $\run$ is bad, the BSCC is not accepting, and so $\candidate(\path_k)$ is bad for every $k \geq 1$. So $\candidate(\path_k)$ is suspect for
%every $k \geq 2$.

\medskip \noindent (b) Suppose the good runs of $\Mcdra$ have nonzero probability.
We construct a finite path, $\path$, starting at $s_0$ so that $\candidate(\path')$ is good for all extensions $\path'$ of~$\path$, and $\candidate(\path'')$ is good or undefined for all prefixes $\path''$ of~$\path$.

Since the good runs of $\Mcdra$ have nonzero probability, $\Mcdra$ has a good BSCC $B$.
Let $\path_1'$ be a simple path from~$s_0$ to a state $s_1 \in B \cap (S \times F)$.
Extend~$\path_1'$ by a shortest path back to $\support{\path_1'}$ (forming a lasso) and denote the resulting path by~$\path_1$.
Observe that $\candidate(\path_1) \subseteq B$ is good, and $\candidate(\path') = \bot$ holds for all proper prefixes $\path'$ of~$\path_1$.
If $\candidate(\path_1) = B$, then we can choose $\path := \path_1$ and $\path$ has the required properties.
Otherwise, let $\path_2'$ be a shortest path extending~$\path_1$ such that $\path_2'$ leads to a state in $B \setminus \candidate(\path_1)$.
Extend that path by a shortest path back to $\candidate(\path_1)$ and denote the resulting path by~$\path_2$.
Then we have $\candidate(\path_1) \subsetneq \candidate(\path_2) \subseteq B$, and $\candidate(\path_2)$ is good, and  $\candidate(\path') \in \{\candidate(\path_1), \bot\}$ holds for all paths~$\path'$ that extend~$\path_1$ and are proper prefixes of~$\path_2$.
Repeat this process until a path $\path$ is found with $\candidate(\path) = B$.
This path has the required properties.
\end{proof}

\paragraph{The cautious monitor.}  The cautious monitor is shown in Algorithm \ref{alg:cautious}. 
The algorithm samples a run of $\Mcdra$ step by step, and resets  whenever the current path $\path$ is bad.
\begin{algorithm}[ht]
  \caption{\textsc{CautiousMonitor}}
  \label{alg:cautious}
\begin{algorithmic}[1]
\While {\textbf{true}}
\State $\path \gets \lambda$ \Comment{Initialize path}
\Repeat
\State $\path \gets \path \concat \nextState(\path)$   \Comment{Extend path}
\Until {$\textsc{Bad}(\path)$} % \textbf{and} $\strength(\path) \geq 2$} \Comment{$\pi$ is suspect} 
\EndWhile
\end{algorithmic}
\end{algorithm}
We formalize its correctness with respect to the specification given in the introduction. Consider the infinite-state Markov chain $\mathcal{C}$ (for cautious) defined as follows. 
The states of $\mathcal{C}$ are pairs $\tup{\path, r}$, where $\path$ is a path of $\Mcdra$, and $r\geq 0$. Intuitively, $j$ counts the number of resets so far. The initial probability distribution assigns probability $1$ to $\tup{\lambda,0}$, and $0$ to all others. The transition probability matrix $\Pm_\mathcal{C}(\tup{\path, r}, \tup{\path',r'})$ is defined as follows.
\begin{itemize}
\item If $\path$ is bad, then 
$$
\Pm_\mathcal{C}(\tup{\path, r}, \tup{\path',r'}) =
\begin{cases}
1 & \mbox{if $\path' = \lambda$ and $r' = r+1$} \\
0 & \mbox{otherwise}
\end{cases}
$$
\noindent We call such a transition a \emph{reset}.
\item If $\path$ is not bad, then
$$
\Pm_\mathcal{C}(\tup{\path, r}, \tup{\path',r}) =
\begin{cases}
\Pm(p, p') & \mbox{if $r'=r$, $\path = \path'' \concat p$ and $\path' = \path \concat p'$} \\
0 & \mbox{otherwise.}
\end{cases}
$$
\end{itemize}
A run of \textsc{CautiousMonitor} corresponds to a run $\run = \tup{\path_1, r_1} \tup{\path_2, r_2} \cdots$ of~$\mathcal{C}$. Let $R$ be the random variable that assigns to $\run$ the supremum of $r_1, r_2 \cdots$. Further, let $S_\varphi$ be the the set of runs such that $R(\run)<\infty$ and the suffix of $\run$ starting immediately after the last reset satisfies $\varphi$.
The following theorem states that \textsc{CautiousMonitor} is correct with respect to the specification described in the introduction. The proof is an immediate consequence of Proposition \ref{prop:suspect}.

\begin{theorem}
Let $\varphi$ be a LTL formula such that $\pr[\mathcal{M} \models \varphi] > 0$.
Let $\mathcal{C}$ be the Markov chain defined as above. We have
\begin{itemize}
\item[(a)] $\pr_\mathcal{C}[R < \infty] = 1$.
\item[(b)] $\pr_\mathcal{C}[S_\varphi  | R < \infty] = \pr_\mathcal{C}[S_\varphi]=1$.
\end{itemize}
\end{theorem}

%\begin{remark}
%If we redefine suspect runs as those in which the current candidate is bad and has strength at least $k$
%for some $k > 2$, Proposition \ref{prop:suspect} still holds, and so the corresponding cautious monitor is still correct. 
%\end{remark}

\paragraph{Performance.} Let $T$ be the random variable that assigns to $\run$ the number of steps till the last reset, or $\infty$ if the number of resets is infinite. First of all, we observe that without any assumption on the system $\expected(T)$ can grow exponentially in the number of states of the chain. Indeed, consider the family of Markov chains of Figure \ref{fig:largescc} and the property $\F p$. Assume the only state satisfying $p$ is $s_{\text{good}}$. Then the product of each chain in the family with the DRA for $\F p$ is essentially the same chain, and the good runs are those reaching $s_{\text{good}}$. We show that even if the controller has full knowledge of the chain $\expected(T)$ grows exponentially. Indeed, since doing a reset brings the chain to $s_0$, it is clearly useless to abort a run that has not yet reached $s_n$. In fact, the optimal monitor is the one that resets whenever the run reaches $s_\text{bad}$. The average number of resets for this controller is clearly 1, and so $\expected(T)$ is the expected number of steps needed to reach $s_\text{bad}$, under the assumption that it is indeed reached. It follows $\expected(T) \ge 2^n$. We formulate this result as a proposition.

\begin{fact}
Let ${\cal M}_n$ be the Markov chain of Figure \ref{fig:largescc} with $n$ states. Given a monitor ${\cal N}$ for the property $\F p$, let  $T_{\cal N}$ be the random variable that assigns to a run of the monitor on  ${\cal M}_n$ the number of steps till the last reset, or $\infty$ if the number of resets is infinite. Then $\expected(T_{\cal N}) \geq 2^n$ for every monitor ${\cal N}$.
\end{fact}

We learn from this example that all monitors have problems when the time needed to traverse a non-bottom SCC of the chain can be very large. So we conduct a parametric analysis in the maximal size $\maxsizescc$ of the SCCs of the chain. This reveals  the weak point of \textsc{CautiousMonitor}: $\expected(T)$ remains exponential even for families satisfying $\maxsizescc=1$. Consider the family of Figure \ref{fig:smallscc}.
\textsc{CautiousMonitor} resets whenever it takes any of the self-loops in states $s_i$. Indeed, after taking a self-loop in state, say $s_i$, the current path $\path$ ends in $s_i s_i$, and so we have $\candidate(\path) = \{s_i\}$, which is a bad candidate. % and $\strength(\path) = 2$.
So after the last reset the chain must follow the path $s_0 s_1 s_2 \cdots s_{\text{good}}$. Since this path has probability $1/2^n$, we get $\expected(T) \geq 2^n$.

In the next section we introduce a ``bold'' monitor. Intuitively, instead of resetting at the first suspicion that the current path may not yield a good run, the bold monitor ``perseveres''.

\section{The bold monitor} 

We proceed in two steps. In Section \ref{subsec:chainwithpmin}, inspired by \cite{DacaHKP17}, we design a bold controller that knows the minimum probability $\pmin$ appearing in 
$\Mc$ (more precisely, a lower bound on it).  In Section \ref{subsec:arbitrarychain} we modify this controller to produce another one that works correctly without any prior knowledge 
about $\Mc$,  at the price of a performance penalty.

\subsection{Chains with known minimal probability}
\label{subsec:chainwithpmin}
The cautious controller aborts a run if the strength of the current candidate exceeds a fixed threshold that remains constant throughout the execution. 
In contrast, the bold controller dynamically increases the threshold, depending on %both the number of resets and on 
the number of different candidates it has seen since the last reset. Intuitively, the controller becomes bolder over time,
which prevents it from resetting too soon on the family of Figure \ref{fig:smallscc}, independently of the length of the chain.
%Continuing with this example, the trick is to choose the rate of growth of the threshold so that the $j$-th run has 
%probability at least $(1 - 2^{-j})$ of reaching the end of the chain, and so of not being aborted (this idea is taken from 
%\cite{DacaHKP17}).
The controller is designed so that it resets almost all bad runs and only a fixed fraction~$\varepsilon$ of the good runs.
Lemma~\ref{lem:bold} below shows how to achieve this. We need some additional definitions.

\paragraph{Strength of a candidate and strength of a path.} Let $\path$ be a path of $\mathcal{M}\otimes \dra$.
The \emph{strength of $\candidate(\path)$} in $\path$ is undefined if $\candidate(\path) = \bot$. 
Otherwise, write $\path = \path' \, s \, \candidatepath$, where $\path'$ is the shortest prefix of $\path$ such that  $\candidate(\path' s)  = \candidate(\path)$; the strength of $\candidate(\path)$ is  the largest $k$ such that every state of $\candidate(\path)$ occurs at least $k$ times in $s \, \candidatepath$, and the last element of $s \, \candidatepath$ occurs at least $k+1$ times. Intuitively, if the strength is $k$ then every state of the candidate has been been exited at least $k$ times but, for technical reasons, we start counting only after the candidate is discovered. The function $\strength(\path)$ returns the strength of $\candidate(\path)$ if $\candidate(\path) \neq \bot$, and $0$ otherwise.

\begin{example}
The following table illustrates the definition of strength.
$$\begin{array}{l|l|l|l|l|c}
\path & \candidate(\path) & \path' & s & \candidatepath & \strength(\path) \\ \hline
p_0p_1 & \bot & - & - & - & 0  \\
p_0p_1p_1 & \{p_1\} & p_0p_1 & p_1 & \epsilon & 0 \\
p_0p_1p_1p_1 & \{p_1\} & p_0p_1 & p_1 & p_1 & 1  \\
p_0p_1p_1p_1p_0 & \{p_0,p_1\} & p_0p_1p_1p_1 & p_0 & \epsilon & 0  \\
p_0p_1p_1p_1p_0p_1& \{p_0,p_1\} & p_0p_1p_1p_1 & p_0 & p_1 & 0 \\
p_0p_1p_1p_1p_0p_1p_0 & \{p_0,p_1\} & p_0p_1p_1p_1 & p_0 & p_1 p_0 & 1 \\
%p_0p_1p_1p_1p_0p_1p_0p_0 & \{p_0,p_1\} & p_0p_1p_1p_1 & p_0 & p_1 p_0 p_0 & 1  \\
p_0p_1p_1p_1p_0p_1p_0p_0 p_1& \{p_0,p_1\} & p_0p_1p_1p_1 & p_0 & p_1 p_0 p_0 p_1& 2 
\end{array}$$
\end{example}

\paragraph{Sequence of candidates of a run.}   Let $\run=s_0s_1\cdots$ be a run of $\Mcdra$. Consider the sequence of random variables defined by $\candidate(s_0\ldots s_j)$ for $j\geq 0$, and let $(\candidate_i)_{i\geq 1}$ be the subsequence without undefined elements and with no repetition of consecutive elements.
For example, for $\varrho=p_0p_1p_1p_1p_0p_1p_2p_2\cdots$, we have $K_1=\{p_1\}$, $K_2=\{p_0,p_1\}$, $K_3=\{p_2\}$, etc.
Given a run $\run$ with a sequence of candidates $\candidate_1, \candidate_2 \ldots, \candidate_k$, we call $\candidate_k$ ithe final candidate. We define the \emph{strength} of $\candidate_i$ in $\run$ as the supremum of the strengths of $\candidate_i$ in all prefixes $\path$ of $\run$ such that $\candidate(\path) = \candidate_i$. %
For technical convenience, we define $\candidate_\ell:=\candidate_j$ for all $\ell>j$ and $\candidate_\infty:=\candidate_j$.  Observe that $\run$ satisfies $\varphi$ if{}f its final candidate is good.  
\begin{lemma} \label{lem:strength-increases}
W.p.1 the final candidate of a run $\run$ is a BSCC of $\Mcdra$. Moreover, for every $k$ there exists a prefix  $\path_k$ of $\run$ such that
$\candidate(\pi_k)$ is the final candidate and $\strength(\pi_k) \geq k$.
\end{lemma}
\begin{proof}
Follows immediately from the definitions, and the fact that w.p.1 the runs of a finite-state Markov chain eventually get trapped in a BSCC and visit
every state of it infinitely often.
\end{proof}

\paragraph{The bold monitor.}  We  bold monitor for chains with minimal probability
$\pmin$, shown in Algorithm \ref{alg:boldalpha}. For every $\run$ and $i \geq 1$, we define two random variables:
\begin{itemize}
\item $\strength_i(\run)$ is the strength of $\candidate _i(\run)$ in $\run$;
\item $\bad_i(\run)$ is \textbf{true} if $\candidate _i(\run)$ is a bad candidate, and \textbf{false} otherwise. 
\end{itemize}
Let $\alpha_0 := \max\{1, - 1/\log(1-\pmin)\}$. The lemma states that, for every $\alpha \ge \alpha_0$ and $\varepsilon > 0$, the runs that satisfy $\varphi$ and in which some bad candidate, say $K_i$, reaches a strength of at least $\alpha(i-\log \varepsilon)$, have probability at most $\varepsilon \pvarphi$. This leads to the following strategy for the controller: when the controller is considering the $i$-th candidate, abort only if the strength reaches $\alpha(i-\log\varepsilon)$.

\begin{lemma}
\label{lem:bold}
Let $\Mc$ be a finite-state Markov chain with minimum probability $\pmin$, let
$\varphi$ be an LTL formula with positive probability $\pvarphi$. %with probability $\pvarphi$ in $\Mc$
For every Markov chain $\Mc$ with minimal probability $\pmin$, for every $\alpha \ge \alpha_0$ and $\varepsilon > 0$: $$\pr \left[ \; \big\{ \run \mid \run \models \varphi \wedge \exists i \geq 1 \, . \, 
\bad_i(\run) \wedge \strength_i(\run) \geq  \alpha(i\rem{+j}\rc{-\log\varepsilon}) \big\}  \; \right] \leq \rem{2^{-j}}\rc{\varepsilon} \pvarphi$$
\end{lemma}
\begin{proof}
	The proof is quite technical and can be found in Appndix~\ref{app:L2} and is inspired by \cite{DacaHKP17}.
	The main technical difficulty, compared to \cite{DacaHKP17} is omnipresent conditioning on the property $\varphi$ being satisfied.
	This also allows for strengthening the bound by the factor of the probability to satisfy it.
\end{proof}

The monitor is parametric in $\alpha$ and~$\varepsilon$. The variable $C$ stores the current candidate, and is used to detect when the candidate changes. The variable $i$ maintains the index of the current candidate, i.e., in every reachable configuration of the algorithm, if $C \neq \bot$ then $C:= \candidate_i$. 

\begin{algorithm}[ht]
  \caption{\textsc{BoldMonitor}$_{\alpha,\epsilon}$}
  \label{alg:boldalpha}
\begin{algorithmic}[1]
\While {\textbf{true}}
\State $\path \gets \lambda$ \Comment{Initialize path}
\State $C \gets \bot$, $i \gets 0$  \Comment{Initialize candidate and candidate counter}
\Repeat
\State $\path \gets \path \concat \nextState(\path)$   \Comment{Extend path}
\If {$\bot \neq \candidate(\path) \neq C$} 
\State $C \gets \candidate(\path)$; $i \gets i+1$ \Comment{Update candidate and candidate counter} 
\EndIf 
\Until {$\textsc{Bad}(\path)$ \textbf{and} $\strength(\path) \geq  \alpha(i\rem{+j}\rc{-\log\varepsilon})$} %\Comment{Reset condition} 
\EndWhile
\end{algorithmic}
\end{algorithm}

The infinite-state Markov chain $\mathcal{B}$ of the bold monitor is defined as the chain $\mathcal{C}$ for the
cautious monitor; we just replace the condition that $\pi$ is bad (and thus has strength at least 1) by the condition that $\candidate(\path)$ is bad and has strength $\alpha(i - \log \varepsilon)$. The random variable $R$ and the event $S_\varphi$ are also defined as for \textsc{CautiousMonitor}.

\begin{theorem}
\label{thm:bold}
Let $\Mc$ be a finite-state Markov chain with minimum probability $\pmin$, and let
$\varphi$ be an LTL formula with probability $\pvarphi > 0$ in $\Mc$.
Let $\mathcal{B}$ be the Markov chain, defined as above, corresponding to the execution of \textsc{BoldMonitor}$_{\alpha,\epsilon}$ on $\Mcdra$, where $\alpha \ge \alpha_0$ and $\varepsilon > 0$. We have:
\begin{itemize}
	\item[(a)] The random variable $R$ is geometrically distributed, with parameter (success probability) at least $\pvarphi(1-\varepsilon)$. Hence, we have $\pr_\mathcal{B}[R < \infty] = 1$ and $\expected_\mathcal{B}(R) \leq %\frac{2}{\pvarphi}$  \todoin{actually$\leq
	 \rc{1/\pvarphi(1-\varepsilon)}$ for every $\varepsilon >0$.
	\item[(b)] $\pr_\mathcal{B}[S_\varphi  | R < \infty] = \pr_\mathcal{B}[S_\varphi]=1$.
\end{itemize}
\end{theorem}

\begin{proof}
%	(a) For every $j > 1$, define the Bernoulli variable $\Reset_j$ as follows:
%	$\Reset_j(\run) = 1$ if{}f the while loop completes $j$ iterations (i.e., in the $j$-th iteration of the while loop the repeat loop terminates). We have:	
%	$$\begin{array}{rll}
%	& \pr \left[\Reset_j = 0 \right]  \\
%	= & 1 - \pr \left[\Reset_j = 1\right] \\
%	= & 1 - (\pr \left[\Reset_j = 1 , \run \models \varphi \right] + \pr \left[\Reset_j =1 , \run \not\models \varphi \right]) \\
%	\geq & 1 - (\rem{2^{-j}}\rc{\varepsilon}\pvarphi + (1 - \pvarphi)) &\text{(by Lemma \ref{lem:bold})}\\
%	\geq & \pvarphi (1- \rem{2^{-j}}\rc{\varepsilon}) 
%	\end{array}$$
%That is, the probability that the $j$-th iteration does not terminate is at least $\pvarphi(1 - \rem{2^{-j}}\rc{\varepsilon})$. \rem{, hence at least $\pvarphi/2$.} 
%So some iteration does not terminate almost surely, which implies $\pr[R < \infty] = 1$.
%
    (a) By Lemma~\ref{lem:strength-increases}, almost all bad runs are reset.
    By Lemma~\ref{lem:bold}, runs, conditioned under being good, are reset with probability at most~$\varepsilon$.
    It follows that the probability that a run is good and not reset is at least $\pvarphi(1-\varepsilon)$.
	(b)~In runs satisfying $R < \infty$, the suffix after the last reset almost surely reaches a BSCC of  $\mathcal{M}\otimes \mathcal{A}$ and visits all its states infinitely often, increasing the strength of the last candidate beyond any bound. So runs satisfying $R < \infty$ belong to $S_\varphi$ with probability 1. 	
%	(c) Each reset has a probability of at least 
%	$\rem{\pvarphi(1-2^{-j}\varepsilon)\geq}\rc{\pvarphi(1-\varepsilon)}$ of being the last reset. 
%	Hence the mean of the respective geometric distribution is \rc{$1/\pvarphi(1-\varepsilon)$.}
\end{proof}

\paragraph{Performance.} Recall that $T$ is the random variable that assigns to a run the number of steps until the last reset.
Let $T_j$ be the number of steps between the $j$-th and $(j+1)$th reset. Observe that all the $T_j$ are identically distributed.
We have  $T_j = T_{j}^\bot + T_{j}^C$, where $T_j^\bot$ and $T_{j}^C$ are the number of prefixes $\path$ such that $K(\path) = \bot$
(no current candidate) and $K(\path) \neq \bot$ (a candidate), respectively.
By deriving bounds on $\expected(T_j^\bot)$ and $\expected(T_j^C)$, we obtain:

\begin{theorem} \label{thm:performance-known-pmin}
Let $\Mc$ be a finite-state Markov chain with $n$ states, minimum probability $\pmin$, and maximal SCC size $\maxsizescc$. Let
$\varphi$ be an LTL formula with probability $\pvarphi > 0$ in $\Mc$.
Let $\alpha \ge \alpha_0$ and $\varepsilon > 0$.
Let $T$ be the number of steps taken by \textsc{BoldMonitor}$_{\alpha,\epsilon}$ until the 
last reset (or $\infty$ if there is no last reset). We have:
\begin{equation}
\label{eq:Tbound6}
\begin{aligned}    
\expected(T)  & \leq  \frac{1}{\pvarphi(1 - \varepsilon)}  \cdot 2 n \alpha  (n-  \log\varepsilon) \maxsizescc \left( \frac{1}{\pmin}\right)^{\maxsizescc}
\end{aligned}
\end{equation}
\end{theorem}

\noindent Here we observe the main difference with \textsc{CautiousMonitor}: Instead of the exponential dependence on $n$ of Theorem \ref{thm:bold}, we only have an exponential dependence on $\maxsizescc$. So for chains satisfying 
$\maxsizescc <\!\!< n$ the bold controller performs much better than the cautious one.

\subsection{General chains}
\label{subsec:arbitrarychain}

We adapt \textsc{BoldMonitor} so that it works for arbitrary finite-state Markov chains, at the price of a performance penalty.
The main idea  is very simple: given any non-decreasing sequence 
$\{\alpha_n\}_{n=1}^\infty$ of natural numbers such that $\alpha_1 = 1$ and $\lim_{n\rightarrow \infty} \alpha_n = \infty$, we sample as in 
\textsc{BoldMonitor}$_{\alpha,\epsilon}$ but, instead of using the same value $\alpha$ for every sample, we use $\alpha_j$ for the $j$-th sample
(see Algorithm~\ref{alg:bold}). The intuition is that $\alpha_j \geq \alpha_0$  holds from some index $j_0$ onwards, and so, by the previous analysis, after the $j_0$-th reset the monitor a.s.\ only executes a finite number of resets. 
Let $\textsc{Sample}(\alpha)$ be the body of the while loop of \textsc{BoldMonitor}$_{\alpha,\varepsilon}$ for a given value of $\alpha$. 
\begin{algorithm}[t]
  \caption{\textsc{BoldMonitor}$_\varepsilon$ for $\{\alpha_n\}_{n=1}^\infty$}
  \label{alg:bold}
\begin{algorithmic}[1]
\State $j \gets 0$ 
\While {\textbf{true}}
\State $j = j+1$
\State \textsc{Sample}$(\alpha_j)$
\EndWhile
\end{algorithmic}
\end{algorithm}
More formally, the correctness follows from the following two properties.
\begin{itemize}
\item For every $j \geq 1$, if \textsc{Sample}$(\alpha_j)$ does not terminate then it executes a good run a.s..\\
Indeed, if \textsc{Sample}$(\alpha_j)$ does not terminate then it a.s.\ reaches a BSCC of $\Mcdra$ and visits all its states infinitely often.  So from some moment on $\candidate(\path)$ is and remains equal to this BSCC, and $\strength(\path)$ grows beyond any bound. Since \textsc{Sample}$(\alpha_j)$ does not terminate, the BSCC is good, and it executes a good run.
\item If $\alpha_j \ge \alpha_0$ then the probability that \textsc{Sample}$(\alpha_j)$ does not terminate is at least \rc{$\varepsilon \pvarphi$}. \\
Indeed, by Lemma \ref{lem:bold}, if $\alpha_j \ge \alpha_0$, the probability is already at least \rc{$\varepsilon \pvarphi$}. Increasing $\alpha$ strengthens the exit condition of the until loop. So the probability that the loop terminates is lower, and the probability of non-termination higher.
\end{itemize} 
These two observations immediately lead to the following proposition:
\begin{proposition}
\label{prop:bold2}
Let $\Mc$ be an arbitrary finite-state Markov chain, and let
$\varphi$ be an LTL formula such that $\pvarphi  := \pr[\mathcal{M} \models \varphi] > 0$.
Let $\mathcal{B}$ be the Markov chain corresponding to the execution of \textsc{BoldMonitor}$_\varepsilon$ on $\Mcdra$ with sequence $\{\alpha_n\}_{n=1}^\infty$.
Let $\pmin$ be the minimum probability of the transitions of $\Mc$ (which is unknown to \textsc{BoldMonitor}$_\varepsilon$). We have
\begin{itemize}
\item[(a)] $\pr_\mathcal{B}[R < \infty] = 1$.
\item[(b)] $\pr_\mathcal{B}[S_\varphi  | R < \infty] = \pr_\mathcal{B}[S_\varphi]=1$.
\item[(c)] $\expected(R) \leq  j_{\text{min}} + \rc{1 / \pvarphi(1 - \epsilon)}$, where $j_{\text{min}}$ is the smallest index $j$ such that 
$\alpha_j \geq \alpha_0$.
\end{itemize}
\end{proposition}

\paragraph{Performance.} Difference choices of the sequence $\{\alpha_n\}_{n=1}^\infty$ lead to versions of \textsc{BoldMonitor}$_\varepsilon$ with different
performance features. Intuitively,  if the sequence grows very fast, then $j_{\text{min}}$ is very small, and the expected number of resets $\expected(R)$ 
is only marginally larger than the number for the case in which the monitor knows $\pmin$. However, in this case
the last $\rc{1 / \pvarphi(1 - \epsilon)}$ aborted runs are performed for very large values $\alpha_j$, and so they take many steps.
If the sequence grows slowly, then the opposite happens; there are more resets, but aborted runs have shorter length. Let us analyze two extreme cases:
$\alpha_j := 2^j$ and $\alpha_j := j$. 

Denote by $f(\alpha)$ the probability that a run is reset, i.e., the probability that a call \textsc{Sample}$(\alpha)$ terminates.
Let further $g(\alpha)$ denote the expected number of steps done in \textsc{Sample}$(\alpha)$ of a run that is reset (taking the number of steps as $0$ if the run is not reset).
According to the analysis underlying Theorem~\ref{thm:performance-known-pmin}, for $\alpha \ge \alpha_0$ we have $g(\alpha) \le c \alpha$ with $c := 2 n  (n-  \log\varepsilon) \maxsizescc\,\pmin^{-\maxsizescc}$.
We can write $T = T_1 + T_2 + \cdots$, where $T_j = 0$ when either the $j$-th run or a previous run is not aborted, and otherwise $T_j$ is the number of steps of the $j$-th run.
For $j \le j_{\text{min}}$ we obtain $\expected(T_j) \le g(\alpha_{j_\text{min}})$ and hence we have:
\begin{align*}
\expected(T) \ &=\ \sum_{j=0}^\infty \expected(T_j) \ \le\ j_{\text{min}} g(\alpha_{j_{\text{min}}}) + \sum_{i=0}^\infty f(\alpha_{j_{\text{min}}})^i g(\alpha_{j_{\text{min}}+i})
\end{align*}
By Theorem~\ref{thm:bold}(a) we have $f(\alpha_{j_{\text{min}}}) \le 1 - \pvarphi(1-\varepsilon)$.
It follows that choosing $\alpha_j := 2^j$ does not in general lead to a finite bound on $\expected(T)$.
Choosing instead $\alpha_j := j$, we get
\begin{align*}
\expected(T) \ &\le \ c j_{\text{min}}^2 + \sum_{i=0}^\infty (1 - \pvarphi(1-\varepsilon))^i c (j_{\text{min}}+i) \\
&\le \ \left( j_{\text{min}}^2 + \frac{j_{\text{min}}}{\pvarphi(1-\varepsilon)} + \frac{1}{(\pvarphi(1-\varepsilon))^2}\right) c\,,
\end{align*}
where $j_\text{min}$ can be bounded by $j_\text{min} \le - 1/\log(1-\pmin) + 1 \le 1/\pmin$.
So with $c = 2 n  (n-  \log\varepsilon) \maxsizescc\,\pmin^{-\maxsizescc}$ we arrive at
\begin{equation}
 \expected(T) \ \le \ \left( \frac{1}{\pmin^2} + \frac{1}{\pmin \pvarphi(1-\varepsilon)} + \frac{1}{(\pvarphi(1-\varepsilon))^2}\right) 2 n  (n-  \log\varepsilon) \maxsizescc\,\pmin^{-\maxsizescc}\,,
\end{equation}
a bound broadly similar to the one from Theorem~\ref{thm:performance-known-pmin}, but with the monitor not needing to know $\pmin$. 
\subsection{Implementing the bold monitor} 
A straightforward implementation of the bold monitor in which the candidate $\candidate(\path)$ and its strength are computed anew 
each time the path is extended is very inefficient. We present a far more efficient algorithm that continuously maintains the candidate of the current 
candidate and its strength. The algorithm runs in $O(n \log n)$ amortized time for a path $\path$ of length $n$, and uses $O(s_n \log s_n)$ space, 
where $s_n$ denotes the number of states visited by $\path$ (which can be much smaller than $n$ when states are visited multiple times). 

Let $\path$ be a path of $\mathcal{M}\otimes \mathcal{A}$, and let $s \in \path$.
(Observe that $s$ now denotes a state of $\mathcal{M}\otimes \mathcal{A}$, not of $\mathcal{M}$.) We let $G_\path = (V_\path, E_\path)$ 
denote the subgraph of $\mathcal{M}\otimes \mathcal{A}$ where $V_\path$ and $E_\path$ are the sets of states and edges visited by $\path$, respectively. 
Intuitively, $G_\path$ is the fragment of $\mathcal{M}\otimes \mathcal{A}$ explored by the path $\path$. We 
introduce some definitions. 
\begin{itemize}
\item The \emph{discovery index} of a state $s$, denoted $d_\path(s)$, is the number of states that appear in the prefix of $\path$ ending 
with the first occurrence
of $s$. Intuitively, $d_\path(s) = k$ if $s$ is the $k$-th state discovered by $\path$. Since different states have different discovery times, 
and the discovery time does not change when the path is extended, we also call $d_\path(s)$ the \emph{identifier} of $s$. 
\item A \emph{root} of $G_\path$ is a state $r \in V_\path$ such that $d_\path(r) \leq d_\path(s)$ for every state  $s \in \scc_\path(r)$,
where $\scc_\path(r)$ denotes the SCC of $G_\path$ containing $s$. Intuitively, $r$ is the first state of $\scc_\path(r)$ visited by $\path$.
\item The \emph{root sequence} $R_\path$ of $\path$ is the sequence of roots of $G_\path$, ordered by ascending discovery index. 
\item Let $R_\path = r_1 \, r_2 \cdots r_m$. We define the sequence $S_\path = S_\path(r_1) \, S_\path(r_2) \cdots S_\path(r_m)$ of sets, where

$$S_\path(r_i) := \{s \in V_\path \mid  d_\path(r_i) \leq d_\path(s) < d_\path(r_{i+1}) \}$$
\noindent for every $1 \leq i < m$, i.e., $S_\path(r_i)$ is the set of states discovered after $r_i$ (including $r_i$) and before $r_{i+1}$ (excluding $r_i$); and
$$S_\path(r_m) := \{s \in V_\path \mid  d_\path(r_m) \leq d_\path(s)\} \ . $$

\item $\mybirthday_\path$ is defined as $\bot$ if $\candidate(\path) = \bot$, and as the length of the shortest prefix $\path'$ of $\path$ such that $\candidate(\path') = \candidate(\path)$ otherwise. Intuitively, $\mathit{Birthday}_\path$ is the time at which the current candidate of $\path$ was created.

\item For every state $s$ of $\path$, let $\path_s$ be the longest prefix of $\pi$ ending at $s$. 
We define $\myvisits_\path(s)$  as the pair $(\mybirthday_{\path_s}, v)$, where $v$ is $0$ if $\mybirthday_{\path_s}=\bot$, and $v$ is the number of times $\pi_s$ has visited $s$ since $\mybirthday_{\path_s}$ otherwise.  We define a total order on these pairs: $(b, v) \preceq (b',v')$ if{}f $b > b'$ (where $\bot > n$ for every number $n$), or $b=b'$ and $v \leq v'$. Observe that, if $\path$ has a candidate, then the smallest pair w.r.t. $\preceq$ corresponds to the state that is among the states visited since the creation of the candidate, and has been visited the least number of times.
\end{itemize}

\noindent The following lemma is an immediate consequence of the definitions.  
\begin{lemma}
Let $G_\path = (V_\path, E_\path)$. The SCCs of $G_\path$ are the sets of $S_\path$. Further, let 
$(b,v)= \min \{ \myvisits_\path(s) \mid s \in V_\path \}$, where the minimum is over $\preceq$. We have $\strength(\path) = v$.
\end{lemma}

By the lemma, in order to efficiently implement \textsc{Monitor} it suffices to maintain $R_\path$, $S_\path$, and 
a mapping $\myvisits_\path$ that assigns $\myvisits_\path(s)$ to each state $s$ of $\path$.
More precisely, assume that \textsc{Monitor} has computed so far a path $\path$ leading to a state $s$, and now it extends $\pi$ 
to $\path' = \path \cdot s'$ by traversing a transition $s \rightarrow s'$ of $\mathcal{M}\otimes \mathcal{A}$;
it suffices to compute $R_{\path'}$, $S_{\path'}$ and  $\myvisits_{\path'}$
from $R_\path$, $S_\path$, and $\myvisits_{\path'}$ in 
$O(\log n)$ amortized time, where $n$ is the length of $\path$. We first show how to update 
$R_\path$, $S_\path$, and $\strength(S_\path)$, and then we describe data structures to maintain them in $O(n \log n)$ amortized time. We consider three cases:
\begin{itemize}
\item $s' \notin V_\path$. That is, the monitor discovers the state $s'$ by traversing $s \rightarrow s'$. Then the SCCs of $G_{\path'}$ are the 
SCCs of $G_\path$, plus a new trivial SCC containing only $s$, with $s$ as root. So $R_{\path'} = R_\path \cdot s$, $S_{\path'} = S_\path \cdot \{s\}$. Since $s'$ has just been discovered, there is no candidate, and so  $\myvisits_{\path'}(s') = (\bot, 0)$, because .
\item $s' \in V_\path$, and $d_\path(s) \leq d_\path(s')$. That is, the monitor had already discovered $s'$, and it had discovered it after $s$. 
Then  $G_{\path'} = (V_\path, E_\path \cup \{(s, s')\}$, but the SCCs of $G_\path$ and $G_{\path'}$ coincide, and so 
$R_{\path'} = R_\path$, $S_{\path'} = S_\path$, and $\strength(S_{\path'}) = \min \{ \#_{\path'}(s'), \strength(S_{\path})\}$.
\item $s' \in V_\path$, and $d_\path(s) > d_\path(s')$. That is, the monitor discovered $s'$ before $s$. Let $R_\path = r_1 \, r_2 \cdots r_m$ and 
let $r_i$ be the root of $\scc_\path(s')$. Then $G_{\path'}$ has a path 
$$r_i \trans{*} r_{i+1} \trans{*} \cdots \trans{*} r_{m-1} \trans{*} r_m \trans{*} s \trans{} s' \trans{*} r_i \ .$$
So we have 
\begin{align*}
R_{\path'} & = r_1 \, r_2 \cdots r_i \\
S_{\path'} & = S_\path(r_1) \cdots S_\path(r_{i-1}) \, \left(\bigcup_{j = i}^m S_{\path}(r_j) \right) 
\end{align*}
Moreover, $\candidate(\path') = S_{\path'}$, because  we have discovered a new candidate. 
Since the strength of a just discovered candidate is $0$ by definition, we set $\strength(\path')=0$. 
\end{itemize}

In order to efficiently update $R_\path$, $S_\path$  and $\min_\path$ we represent them using the following data structures.
\begin{itemize}
\item The number $N$ of different states visited so far.
\item A hash map $D$ that assigns to each state $s$ discovered by $\path$ its discovery index. When
$s$ is visited for the first time, $D(s)$ is set to $N+1$. Subsequent lookups return $N+1$.
\item A structure $R$ containing the identifiers of the roots of $R_\path$, and supporting the following operations in
amortized $O(\log n)$ time: $\textsc{insert}(r)$, which inserts the identifier of $r$ in $R$; $\textsc{extract-max}$, which returns the largest identifier
in $R$;  and $\textsc{find}(s)$, which returns the largest identifier of $R$ smaller than or equal to the identifier of $s$. (This is the identifier of the root of
the SCC containing $s$.) For example, this is achieved by implementing $R$ both as a search tree and a heap. 
\item For each root $r$ a structure $S(r)$ that assigns to each state of $s \in S(r)$ the value $\myvisits_\path(s)$,
and supports the following operations in amortized $O(\log n)$ time: $\textsc{find-min}$, which returns the minimum value of the states of $S(r)$; $\textsc{increment-key}(s)$, which increases the value of $s$ by 1, and $\textsc{merge}$, which returns the union of two given maps.
For example, this is achieved by implementing $S(r)$ as a Fibonacci heap. 
\end{itemize}

When the algorithm explores an edge $s \rightarrow s'$ of the Markov chain $\mathcal{M}\otimes \mathcal{A}$, it
updates these data structures as follows. The algorithm first computes $D(s')$, and then proceeds according to the three cases above:
\begin{itemize}
\item[(1)] If $s'$ had not been visited before (i.e., $D(s') = N+1$), then the algorithm sets $N := N+1$, inserts $D(s')$ in $R$, and creates a new Fibonacci heap $S(s')$ containing only the state $s'$ with key 1.
\item[(2)] If $s'$ had been visited before (i.e., $D(s') \leq  N$), and $D(s) \leq D(s')$, then the algorithm executes
$\textsc{find}(s)$ to find the root $r$ of the SCC containing $s$,  and then increments the key of $s$ in $S(r)$ by 1.
\item[(3)] If $s'$ had been visited before (i.e., $D(s') \leq  N$), and  $D(s) > D(s')$, then the algorithm 
executes the following pseudocode, where $\sigma$ is an auxiliary Fibonacci heap:
\begin{center}
\begin{minipage}{5cm}
\begin{algorithmic}[1]
\State $\sigma \gets \emptyset$
\Repeat 
\State  $r \gets \textsc{extract-max}(R)$
\State $\sigma \gets \textsc{merge}(\sigma, S(r))$ 
\Until{$D(r) \leq  D(s')$}  
\State $\textsc{insert}(r, R)$
\end{algorithmic}
\end{minipage}
\end{center}
\end{itemize}
At every moment in time the current candidate is the set $S(r)$, where $r = \textsc{extract-max}(R)$, and its strength can be obtained from $\textsc{find-min}(S(r))$.

Let us now examine the amortized runtime of the implementation.  Let $n_1, n_2, n_3$ be the number of steps executed by the algorithm corresponding to the cases (1), (2), and (3) above.  In cases (1) and (2), the algorithm executes a constant
number of heap operations per step, and so it takes $O((n_1 + n_2) \log n)$ amortized time for all steps together. 
This is no longer so for case (3) steps. For example, if the Marvov chain is a big elementary circuit $s_0 \trans{} s_1 \trans{} \cdots \trans{} s_{n-1} \trans{} s_0$, then at each step but the last one we insert one state into the heap, and at the last step we extract them all; that is, the last step takes $O(n)$ heap operations. However,  observe that each state is inserted in the heap exactly once, when it is discovered, and extracted at most once. So the algorithm executes at most $n$
 \textsc{extract-max} and \textsc{merge} heap operations for all case (3) steps together, and the amortized time over all of them is $O(n_3 \log n)$. This gives an overall amortized runtime of $O(n \log n)$.

%The pseudocode is shown in 
%
%\begin{algorithm}[ht]
%  \caption{\textsc{ImpMonitor}}
%  \label{alg:ltl}
%\begin{algorithmic}[1]
%\Require Markov Chain $\mathcal{M}\otimes \mathcal{A}$
%\State $j \gets 0$ \Comment{Initialize iteration counter}
%\While {\textbf{true}}
%\State $j \gets j+1$ \Comment{Update iteration counter}
%\State $\path \gets \lambda$ \Comment{Initialize path}
%\State $C \gets \bot$, $i \gets 0$  \Comment{Initialize candidate and candidate counter}
%\State $N= 0$, $R = \epsilon$, $H = \bot$
%\Repeat
%\State $s \gets \nextState$   
%\If {$s \notin H$} $N \gets N+1$; $D(s) \gets N+1$; push $(s, R)$; $S(s) \gets \{ (s, 1) \}$; 
%\Else $\textit{OldC} \gets C$, $C \gets \candidate(\path)$
%\Comment{Update candidate}
%\If {$\bot \neq C\neq \textit{OldC}$} $i \gets i+1$ \Comment{Update candidate counter} \EndIf 
%\Until {$\textsc{Bad}(C)$ \textbf{and} $\strength(\path) \geq  \alpha(i+j)$}
%\EndWhile
%\end{algorithmic}
%\end{algorithm}

\section{Experimental results}

In this section, we illustrate the behaviour of our monitors on concrete models from the PRISM Benchmark Suite \cite{DBLP:conf/qest/KwiatkowskaNP12} and compare the obtained theoretical bounds to the experimental values.
To this end, we have re-used the code provided in \cite{DacaHKP17}, which in turn is based on the PRISM  model checker \cite{prism}.
Whenever the obtained candidate is actually an accepting BSCC, we have a guarantee that no restart will ever happen and we terminate the experiment.

Table~\ref{tab:exper} shows the results on several models.
For the bluetooth benchmark, the optimal number of restarts is $1/0.2=5$ and with $\varepsilon=0.5$ it should be smaller than 10. 
We see that while the bold monitor required $6.6$ on average, the cautious one indeed required a bit more.
For Hermann's stabilization protocol, almost all sufficiently long runs have a good candidate.
In this model, we have not even encountered any bad candidate on the way. 
This can be easily explained since only 38 states out of the half million are outside of the single BSCC. 
They are spread over 9 non-bottom SCCs and some states are transient; however, no runs got stuck in any of the tiny SCCs.
A similar situation occurs for the case of gridworld, where on average every tenth run is non-satisfying. 
However, in our five repetitions (each with a single satisfying run), we have not encountered any.
Finally,  we could not determine the satisfaction probability in crowds since PRISM times out on this model with more than two million states.
However, one can still see the bold monitor requiring slightly less resets than the cautious one, predicting $\pvarphi$ to be in the middle of the $[0,1]$-range.
It is also worth mentioning that the large size preventing a rigorous numeric analysis of the system did not prevent our monitors from determining satisfaction on single runs.

\newcommand{\gf}{\G\F}
\newcommand{\fg}{\F\G}
\begin{table}[t]
	\caption{Experimental comparison of the monitors, showing the average number of restarts and average length of a restarted run. The average is taken over five runs of the algorithm and $\varepsilon=0.5$.}
	\label{tab:exper}
\begin{tabular}{llr@{\hskip 2mm}lr@{\hskip 4mm}lrr}
Model &Property &States&$\pmin$&$\pvarphi$&Monitor&Avg. $R$&Avg. $\frac T R$\\\hline
bluetooth&\gf&143,291&$0.008$&0.20&Cautious &9.0 & 4578\\
&&&&&Bold&6.6&3758\\
hermann&\fg&524,288&$1.9\cdot 10^{-6}$&1&Cautious&0&-\\
&&&&&Bold&0&-\\
gridworld&$\gf\to\fg$& 309,327&0.001& 0.91& Cautious & 0 & -\\
&&&&&Bold&0&-\\
crowds& \fg&2,464,168& 0.066&?&Cautious& 0.8 & 63\\
&&&&&Bold&0.6&90\\
\end{tabular}
\end{table}
\section{Conclusions}
We have shown that monitoring of arbitrary $\omega$-regular properties is possible for finite-state Markov chains, even if the monitor has no information \emph{at all} about the chain, its probabilities, or its structure. More precisely, we have exhibited
monitors  that ``force'' the chain to execute runs satisfying a given property $\varphi$ (with probability 1). The monitors reset the chain whenever the current run is suspect of not satisfying $\varphi$. They work even if $\varphi$ is a liveness property without any ``good prefix'', i.e., a prefix after which any extension satisfies $\varphi$.

Unsurprisingly, the worst-case behaviour of the monitor, measured as the number of steps until the last reset, 
is bad when the probability of the runs satisfying $\varphi$ or the minimal probability of the transitions of the chain are very small, or when the strongly connected components of the chain are very large. We have given performance estimates that quantify the relative weight of each of these parameters.  The design of dedicated monitors that exploit information on this parameters is an interesting topic for future research.

\myspace
\bibliographystyle{plain}
\bibliography{ref}

\newpage
\appendix
\section{Technical Proofs}

\subsection{Proof of Lemma~\ref{lem:bold}}\label{app:L2}

\newcommand{\sat}{\candidate_\infty\models\varphi}
\newcommand{\myexit}{t\to t'}
	
Let $\bscc$ denote the set of BSCCs of the chain-automaton product and $\scc$ the set of its SCCs.

For a subset $K$ of states of the product, $\skcand k(\candidate)$ denotes the event (random predicate) of $\candidate$ being a candidate with strength at least $k$ on a run of the product.
Further, the ``weak'' version $\kcand k(\candidate)$ denotes the event that $\candidate$ has strength $k$ when counting visits even prior to discovery of $\candidate$, i.e. each state of $\candidate$ has been visited and exited at least $k$ times on a prefix $\path$ of the run with $\candidate(\path)=K$.
Previous work bounds the probability that a non-BSCC can be falsely deemed BSCC based on the high strength it gets.

\begin{lemma}[\cite{DacaHKP17}]%\label{lem:one-cand}
	For every set of states $\candidate\notin\bscc$, and every $s\in \candidate$,  $k\in\Nset$,
	$$
	\pr_s[\kcand k (\candidate)]\leq (1-\pmin)^k\,.
	$$
\end{lemma}
\begin{proof}
	Since $\candidate$ is not a BSCC, there is a state $t\in\candidate$ with a transition to $t'\notin\candidate$.
	The set of states $\candidate$ becomes a $k$-candidate of a run starting from $s$, only if $t$ is visited at least $k$ times by the path and was never followed by $t'$ (indeed, even if $t$ is the last state in the path, by definition of a $k$-candidate, there are also at least $k$ previous occurrences of $t$ in the path). 
	Further, since the transition from $t$ to $t'$ has probability at least $\pmin$, the probability of not taking the transition $k$ times is at most $(1-\pmin)^k$. 
\end{proof}

In contrast to \cite{DacaHKP17}, we need to focus on runs where $\varphi$ is satisfied.
For clarity of notation, we let $\candidate\models\varphi$ denote that $\candidate$ is good, and $\candidate\not\models\varphi$ denote that $\candidate$ is bad.
In particular, $\sat$ %$\candidate_\infty\models\varphi$ 
denotes the event that the run satisfies $\varphi$.

\begin{lemma}\label{lem:one-cand}
	For every set of states $\candidate\notin\bscc$, and every $s\in \candidate$,  $k\in\Nset$,
	$$
	\pr_s[\kcand k (\candidate)\mid \sat]\leq (1-\pmin)^k\,.
	$$
\end{lemma}
\begin{proof}
	The previous argument applies also in the case where we assume that after this strength is reached the run continues in any concrete way (also satisfying $\varphi$) due to the Markovian nature of the product:

	\begin{align*}
	  &\pr_s[\kcand k (\candidate)\mid \sat]\\
	=& \sum_{\myexit}\pr_s[\kcand k (\candidate), K\text{ exited by }\myexit\mid \sat]\\
	=& \sum_{\myexit}\pr_s[\kcand k (\candidate), K\text{ exited by }\myexit, \sat]/\pr_s[ \sat]\\
	=& \sum_{\myexit}\pr_s[\kcand k (\candidate), K\text{ exited by }\myexit] \cdot \pr_s[\sat\mid \kcand k (\candidate), K\text{ exited by }\myexit] /\pr_s[ \sat]\\
	\stackrel{(1)}=& \sum_{\myexit}\pr_s[\kcand k (\candidate), K\text{ exited by }\myexit] \cdot\pr_s[\sat\mid K\text{ exited by }\myexit] /\pr_s[\sat]\\
	\stackrel{(2)}=& \sum_{\myexit}\pr_s[\kcand k (\candidate), K\text{ exited by }\myexit] \cdot\pr_{t'}[\sat] /\pr_s[ \sat]\\
	\leq&\sum_{\myexit\text{ exiting }K}\pr_s[\text{reach }t]\pr_t[\text{not take }\myexit\text{in }k\text{ visits of }t]\cdot\Pm(t,t') \cdot\pr_{t'}[\sat] /\pr_s[ \sat] \\
	=&\sum_{\myexit\text{ exiting }K}\pr_t[\text{not take }\myexit\text{in }k\text{ visits of }t] \pr_s[\text{reach }t]\cdot\Pm(t,t')\cdot\pr_{t'}[\sat] /\pr_s[ \sat] \\
	\leq&\sum_{\myexit\text{ exiting }K}(1-\pmin)^k\pr_s[\text{reach $t'$ as the first state outside $K$}] \cdot\pr_{t'}[\sat] /\pr_s[ \sat]\\
	=& (1-\pmin)^k \pr_s[\sat] /\pr_s[ \sat]\\
	=&(1-\pmin)^k
	\end{align*}

\noindent where (1) follows by the Markov property and by a.s. $\candidate\neq\candidate_\infty$, (2) by the Markov property.
\end{proof}

In the next lemma, we lift the results from fixed designated candidates to arbitrary discovered candidates, at the expense of requiring the (strong version of) strength instead of only the weak strength.
To that end, let \emph{birthday} $b_i$ be the moment when $i$th candidate on a run is discovered, i.e., a run is split into $\run=\path b_i\run'$ so that $\candidate_i=\candidate(\path b_i)\neq\candidate(\path)$.
In other terms, $b_i$ is the moment we start counting the occurences for the strength, whereas the weak strength is already 1 there.

%\begin{lemma}[known]\label{lem:ith-cand}
%	For every $i,k\in\Nset$, we have $$\pr[\skcand k(\candidate_i) \mid \candidate_i\notin\bscc]\leq (1-\pmin)^k\,.$$ 
%	%\todoprzemek{$\pr[\candidate_i\in\scs\setminus \bscc,b_i\in K_i]=\pr[\candidate_i\notin\bscc]$ instead?}
%\end{lemma}
%\begin{proof}
%	\begin{align*}
%	&\pr[\skcand k(\candidate_i) \mid \candidate_i\notin\bscc]\\
%	&= \frac{\pr[\skcand k(\candidate_i), \candidate_i\notin\bscc]}{\pr[\candidate_i\notin\bscc]}\\
%	&= \frac1{\pr[\candidate_i\notin\bscc]} \sum_{\substack{C\in\scs\setminus\bscc\\s\in C}}\pr[\skcand k(C),\candidate_i=C,b_i=s]\\
%	&= \frac1{\pr[\candidate_i\notin\bscc]} \sum_{\substack{C\in\scs\setminus\bscc\\s\in C}}\pr[\candidate_i=C,b_i=s]\pr_s[\kcand k (C)] \tag{by Markov property}\\
%	&\leq \frac1{\pr[\candidate_i\notin\bscc]} \sum_{\substack{C\in\scs\setminus\bscc\\s\in C}}\pr[\candidate_i=C,b_i=s](1-\pmin)^k \tag{by Lemma~\ref{lem:one-cand}}\\
%	&= (1-\pmin)^k~\tag{since $\pr[\candidate_i\notin\bscc]=\sum_{\substack{C\in\scs\setminus\bscc\\s\in C}}\pr[\candidate_i=C,b_i=s]$}.\\
%	\end{align*}
%\end{proof}

\begin{lemma}\label{lem:ith-cand}
	For every $i,k\in\Nset$, we have $$\pr[\skcand k(\candidate_i) \mid \candidate_i\notin\bscc,\sat]\leq (1-\pmin)^k\,.$$ 
%	and hence
%	$$\pr[\kcand {k+1}(\candidate_i) \mid \candidate_i\notin\bscc,\sat]\leq (1-\pmin)^k\,.$$\todo{Cand not true, only SCand}
	%\todoprzemek{$\pr[\candidate_i\in\scs\setminus \bscc,b_i\in K_i]=\pr[\candidate_i\notin\bscc]$ instead?}
\end{lemma}
\begin{proof}
	\begin{align*}
	&\pr[\skcand k(\candidate_i) \mid \candidate_i\notin\bscc,\sat]\\[0.2cm]
	&= \frac{\pr[\skcand k(\candidate_i), \candidate_i\notin\bscc,\sat]}{\pr[\candidate_i\notin\bscc,\sat]}\\
	&= \frac1{\pr[\candidate_i\notin\bscc,\sat]} \sum_{\substack{C\in\scs\setminus\bscc\\s\in C}}\pr[\skcand k(C),\candidate_i=C,b_i=s,\sat]\\
	&= \frac1{\pr[\candidate_i\notin\bscc,\sat]} \sum_{\substack{C\in\scs\setminus\bscc\\s\in C}}\pr[\candidate_i=C,b_i=s]\cdot\pr_s[\kcand k (C),\sat] \\
	&= \frac1{\pr[\candidate_i\notin\bscc,\sat]} \sum_{\substack{C\in\scs\setminus\bscc\\s\in C}}\pr[\candidate_i=C,b_i=s]\cdot\pr_s[\kcand k (C)\mid\sat]\cdot\pr_s[\sat] \\
	&\leq \frac{(1-\pmin)^k }{\pr[\candidate_i\notin\bscc,\sat]} \sum_{\substack{C\in\scs\setminus\bscc\\s\in C}}\pr[\candidate_i=C,b_i=s]\cdot\pr_s[\sat]\tag{by Lemma~\ref{lem:one-cand}}\\
	&\leq \frac{(1-\pmin)^k }{\pr[\candidate_i\notin\bscc,\sat]} \sum_{\substack{C\in\scs\setminus\bscc\\s\in C}}\pr[\candidate_i=C,b_i=s,\sat]\\	
	&= (1-\pmin)^k
	\end{align*}
	with the last equality due to $$\candidate_i\notin\bscc\cap\sat=\biguplus_{\substack{C\in\scs\setminus\bscc\\s\in C}}\candidate_i=C,b_i=s,\sat$$
\end{proof}

The set $\mathcal E\mathit{rr}$ of the next lemma is actually exactly the set considered in Lemma~\ref{lem:bold} but in a more convenient notation for the computation.

\begin{lemma}\label{lem:more-cand}
	For $(k_i)_{i=1}^\infty\in\Nset^\Nset$, let 
	$\mathcal E\mathit{rr}$ be the set of runs such that for some $i\in\Nset$, we have $\skcand {k_i}(\candidate_i)$ despite $\candidate_i\not\models\varphi$ and $\candidate_{\infty}\models\varphi$.
	Then $$\displaystyle\pr[\mathcal E\mathit{rr}]<p_\varphi\sum_{i=1}^\infty (1-\pmin)^{k_i}\,.$$
\end{lemma}
\begin{proof}
	\begin{align*}
	\pr[\mathcal E\mathit{rr}]
	&=\pr\left[\bigcup_{i=1}^{\infty} \Big( \skcand {k_i}(\candidate_i) \cap \candidate_i\not\models\varphi\cap\sat\Big)\right]\\
	&\leq\pr\left[\bigcup_{i=1}^{\infty} \Big( \skcand {k_i}(\candidate_i) \cap \candidate_i\notin\bscc\cap\sat\Big)\right]\\
	&\leq\sum_{i=1}^{\infty}\pr[\skcand {k_i}(\candidate_i)\cap \candidate_i\notin\bscc \cap \sat] \tag{by the union bound}\\
	&=\sum_{i=1}^{\infty}\pr[\skcand {k_i}(\candidate_i) \mid \candidate_i\notin\bscc\cap\sat]\cdot\pr[\candidate_i\notin\bscc\mid\sat]\cdot\pr[\sat]\\
	&\leq \sum_{i=1}^{\infty}\pr[\skcand {k_i}(\candidate_i) \mid \candidate_i\notin\bscc\cap\sat]\cdot 1\cdot p_\varphi\\
	&=p_\varphi\sum_{i=1}^{\infty}(1-\pmin)^{k_i}~. \tag{by Lemma~\ref{lem:ith-cand}}
	\end{align*}	
%	$\candidate_i\not\models\varphi\cap\candidate_{\infty}\models\varphi \subseteq \candidate_i\neq\candidate_\infty\subseteq \candidate_i\notin\bscc$ and conclude by previous lemma.
\end{proof}

%\begin{lemma}[`L.2']
%	For every Markov chain $\Mc$ with minimal probability $\pmin$, for every $\varepsilon>0$: $$\pr \left[ \; \big\{ \run \mid \run \models \varphi \wedge \exists i \geq 1 \, . \, 
%	\candidate_i(\run)\not\models\varphi \wedge \skcand {k_i}(\candidate_i) \big\}  \; \right] \leq p_\varphi \cdot \varepsilon \ .$$
%	where $k_i=(i+\log\varepsilon)/\log(1-\pmin)$.
%\end{lemma}
\begin{proof}[of Lemma~\ref{lem:bold}]
%	$$\pr \left[ \; \big\{ \run \mid \run \models \varphi \wedge \exists i \geq 1 \, . \, 
%	\candidate_i(\run)\not\models\varphi \wedge \skcand {k_i}(\candidate_i) \big\}  \; \right] \leq p_\varphi \cdot \varepsilon \ .$$
%	where $k_i=(i+\log\varepsilon)/\log(1-\pmin)$.
	Recall that Lemma~\ref{lem:bold} claims that 
	$$\pr[\mathcal E\mathit{rr}]\leq \rc{\varepsilon}p_\varphi$$
	for $k_i:=(i\rc{-\log\varepsilon})\cdot\frac{-1}{\log(1-\pmin)}$.
	Directly from previous lemma by plugging in these $k_i$, we obtain  
	$$\pr[\mathcal E\mathit{rr}]\leq p_\varphi\sum_{i=1}^\infty (1-\pmin)^{k_i}=p_\varphi\sum_{i=1}^\infty 2^{-i}\rc{2^{\log\varepsilon}}=p_\varphi \rc{\varepsilon}\,.$$
\end{proof}

\subsection{Proof of Theorem~\ref{thm:performance-known-pmin}}\label{app:performance-known-pmin}

\begin{proof}[of Theorem~\ref{thm:performance-known-pmin}]
Let $I_{ji,k}$ be the number of steps between the $j$-th and $(j+1)$th reset such that the current candidate is $K_i$, and its strength is 
$k$. Observe that for a Markov chain with $n$ states we have $I_{ji,k} = 0$ if $i > n$ or $j > \alpha(i- \log\varepsilon)$. Indeed, if the Markov chain has $n$ states, 
then along the run there are at most $n$ candidates; moreover, the strength of the $K_i$ stays strictly below
$\alpha(i-\log\varepsilon)$, because otherwise the run is aborted.  So we have

\begin{equation}
\label{eq:T}
T = \sum_{j=1}^\infty T_j = \sum_{j=1}^\infty T_{j}^\bot + \sum_{j=1}^\infty T_{j}^C = \sum_{j=1}^\infty T_{j}^\bot + 
\sum_{j=1}^\infty \sum_{i=1}^n   \sum_{k=1}^{\alpha(i- \log\varepsilon)}  I_{ji,k}
\end{equation}

\noindent and so, by linearity of expectations,

\renewcommand{\arraystretch}{1.5}
\begin{equation}
\label{eq:Tbound}
\begin{aligned}
\expected(T)  & =    \expected\left( \sum_{j=1}^{\infty} \left( T_{j}^\bot + \sum_{i=1}^n \sum_{k=1}^{\alpha(i- \log\varepsilon)} I_{ji,k} \right) \right)  \\
                    & =    \expected\left( \sum_{j=1}^{\infty} T_{j}^\bot  \right) + 
                                 \sum_{i=1}^n   \sum_{k=1}^{\alpha(i- \log\varepsilon)}  \sum_{j=1}^{\infty} \expected\left( I_{ji,k} \right)  
\end{aligned}
\end{equation}
Let us bound the first summand. Since  $\candidate(\pi)= \bot$ only holds when the 
last state of $\path$ is visited for the first time, we have $T_j^\bot \leq n$. Moreover, $T_j^\bot = 0$ for every $j \geq R$, the number of resets.
So we get
\begin{equation}
\label{eq:Tbound2}
\expected\left( \sum_{j=1}^{\infty} T_{j}^\bot  \right) \leq \expected(n \cdot R) = n \cdot \expected(R) 
\end{equation}
Consider now the variables $I_{ji,k}$. If $j \geq R$ then $I_{ji,k}=0$ by definition, since there is no $(j+1)$-th reset. Moreover, under the condition
$j < R$ the variables $I_{ji,k}$ and $I_{(j+1)i,k}$ have the same expectation, because they refer to different runs. By Theorem \ref{thm:bold}(a)  $R$ is geometrically distributed
with parameter at least $\pvarphi(1-\varepsilon)$, and so we get
\begin{equation}
\label{eq:Tbound3}
\expected(I_{(j+1)i,k}) \le \expected(I_{ji,k}) \cdot (1-\pvarphi(1-\varepsilon))
\end{equation}
Plugging (\ref{eq:Tbound}) and (\ref{eq:Tbound2}) into (\ref{eq:T}), and taking into account that $\expected(R) \leq 1/ \pvarphi (1 - \varepsilon)$, we obtain

\begin{equation}
\label{eq:Tbound4}
\begin{aligned}    
\expected(T)  & \leq \expected(n \cdot R) +    \sum_{i=1}^n   \sum_{k=1}^{\alpha(i- \log\varepsilon)}  \left(\expected(I_{0i,k}) \sum_{j=0}^{\infty}  (1-\pvarphi(1-\varepsilon))^j \right) \\
                    & = n \cdot \expected(R)    +    \sum_{i=1}^n   \sum_{k=1}^{\alpha(i- \log\varepsilon)}  \frac{\expected(I_{0i,k})}{\pvarphi(1- \varepsilon)} \\
                    & \leq  \frac{1}{\pvarphi(1 - \varepsilon)} \left( n + \sum_{i=1}^n   \sum_{k=1}^{\alpha(i- \log\varepsilon)}  \expected(I_{0i,k})\right)
\end{aligned}
\end{equation}
If we can find an upper bound $I \geq \expected(I_{0i,k})$ for every $i, k$, then we finally get:

\begin{equation}
\label{eq:Tbound5}
\begin{aligned}    
\expected(T)  & \leq  \frac{1}{\pvarphi(1 - \varepsilon)}  \cdot n \cdot  \left( 1 +    \alpha (n-  \log\varepsilon)  \cdot I \right) \\
& \leq \frac{1}{\pvarphi(1 - \varepsilon)}  \cdot 2 n \alpha (n-  \log\varepsilon)  I
\end{aligned}
\end{equation}

Before estimating the bound $I$ let us consider the family of chains of Figure \ref{fig:smallscc}, 
and the property $\F p$. In this case the candidates contain only one state, and their strength increases whenever the self-loop
on the state is traversed.  So $\expected(I_{0i,k}) \leq  1$ holds for every $i, k$, and so we can take $I := 1$.
%Moreover, for this family $\pmin = 1/2$ and $\pvarphi=1$. Plugging these values in we obtain $\expected(T) \leq (n / \pvarphi(1 - \varepsilon)) (n + 1 -\log \varepsilon)$.

We now compute a bound $I \geq \expected(I_{0i,k})$ valid for arbitrary chains. Recall that $\expected(I_{0i,k})$ is the number of steps it takes to increase the 
strength of the $i$-th candidate $\candidate_i$ of the $0$-th run from $k$ to $k+1$. This is bounded by the number of steps it takes to visit every state of $\candidate_i$ once.  Let $\maxsizescc \in O(n)$ be the maximal  
size of a SCC. Given any two states $s, s'$ of an SCC, the probability of reaching $s'$ from $s$ after at most $\maxsizescc$ steps is at least $\pmin^{\maxsizescc}$. 
So the expected time it takes to visit every state of an SCC at least once is bounded by $\maxsizescc \cdot \pmin^{-\maxsizescc}$. So taking $I := \maxsizescc \cdot \pmin^{-\maxsizescc}$
we obtain the final result.
\qed
\end{proof}

\end{document}